\documentclass[11pt,a4paper]{article}

\usepackage[T1]{fontenc}
\usepackage{lmodern}
\usepackage{mathrsfs}
\usepackage{pdfpages}
\usepackage[margin=2.6cm]{geometry}
\usepackage{amsmath,amssymb,amsthm,mathtools}
\usepackage{graphicx}
\usepackage{booktabs}
\usepackage{array}
\usepackage{bm}
\usepackage{xcolor}
\usepackage{microtype}
\usepackage[colorlinks=true,
            linkcolor=blue!55!black,
            citecolor=blue!55!black,
            urlcolor=blue!55!black]{hyperref}

\numberwithin{equation}{section}

\newcommand{\RR}{\mathbb R}
\newcommand{\CC}{\mathbb C}
\newcommand{\ZZ}{\mathbb Z}
\newcommand{\QQ}{\mathbb Q}

\newcommand{\cE}{\mathcal E}
\newcommand{\cI}{\mathcal I}
\newcommand{\cP}{\mathcal P}

\newtheorem{theorem}{Theorem}[section]
\newtheorem{proposition}[theorem]{Proposition}
\newtheorem{corollary}[theorem]{Corollary}

\hypersetup{
  pdftitle={Superintegrability and Choreographic Obstructions in Dihedral n-Body Hamiltonian Systems},
  pdfauthor={Adrian M. Escobar-Ruiz and Manuel Fernandez-Guasti}
}

\begin{document}

\title{\bfseries Superintegrability and choreographic obstructions\\
in dihedral $n$-body Hamiltonian systems}
\author{A. M. Escobar-Ruiz \and M. Fernandez-Guasti}
\date{}
\maketitle

\begin{center}
\small
Departamento de F\'isica, Universidad Aut\'onoma Metropolitana--Iztapalapa,\\
Rafael Atlixco 186, 09340 Ciudad de M\'exico, Mexico\\[2pt]
\href{mailto:admau@xanum.uam.mx}{\texttt{admau@xanum.uam.mx}},
\href{mailto:mfg@xanum.uam.mx}{\texttt{mfg@xanum.uam.mx}}
\end{center}

\begin{abstract}
We study planar $n$-body Hamiltonians with quadratic pair interactions
whose coupling pattern is invariant under the dihedral group $D_n$.  After the
center of mass is removed, the dynamics separates into label-Fourier modes.
This makes it possible to distinguish three notions that are often conflated:
commensurability of the normal frequencies, maximal superintegrability of the
full relative Hamiltonian, and realization of choreographic space--time
symmetry.  We prove that a periodic configuration is $C_n$-equivariant if and
only if every label-Fourier coefficient acquires under the time shift $T/n$ the
character phase prescribed by cyclic relabeling.  For harmonic motion this
criterion determines the complete equivariant subspace of initial data and its
codimension.  Thus a resonant orbit may be periodic without being a
choreography.  The cases $n=4,5,6$ illustrate the obstruction.  In particular,
$n=6$ is the first case with three inequivalent internal branches: the
nondegenerate $1{:}2{:}3$ resonance admits a phase-matched three-branch
choreography, whereas the exactly degenerate $1{:}2{:}2$ spectrum is
incompatible with simultaneous $m=2$ and Nyquist $m=3$ excitation.  A
choreography on the latter surface survives only in a compatible invariant
subspace.  Resonant data outside the equivariant subspace can instead organize
into synchronized multi-trace fragments.
\end{abstract}

\noindent\textbf{Keywords:} choreographies; superintegrable systems; dihedral
symmetry; Fourier characters; resonant oscillators; equivariant dynamics.

\medskip
\noindent\textbf{2020 Mathematics Subject Classification:}
70H06; 37J35; 70F10; 20C15.

\section{Introduction}\label{sec:introduction}

A planar $n$-body choreography is a collision-free periodic motion in which
identical particles follow one closed curve with equal time delays.  Since the
figure-eight solution, choreographies have been studied by variational,
topological, numerical, and symmetry-based methods
\cite{moore1993braids,chenciner2000remarkable,simo2002dynamical,
ferrario2004relative,montgomery1998braid,kapela2007symbolic,Yu2017}.
The possible space--time symmetry groups and the topology of the associated
loop spaces have also been classified
\cite{MontaldiSteckles2013,BarutelloFerrarioTerracini2011}.  Explicit analytic
families remain comparatively rare
\cite{fujiwara2003lemniscate,fernandez2025fourbody,fernandez2025nbody}.

The literature distinguishes simple choreographies, in which every body
traverses one time-shifted trace, from periodic organizations involving more
than one choreographic trace
\cite{ChencinerGerverMontgomerySimo2002,barutello2006double}.  This distinction
motivates our use of \emph{fragmentation} for synchronized multi-trace motion;
the examples below are exact harmonic solutions and are not obtained through
an action-minimization argument.

Here we use an exactly solvable class of quadratic $D_n$-invariant Hamiltonians
to isolate a representation-theoretic obstruction to choreography.  The
dihedral symmetry acts on particle labels: the coupling assigned to a pair
depends only on its separation on an abstract regular $n$-gon.  It does not
require the instantaneous configuration in the physical plane to be a regular
polygon.  The label-space quadratic form is circulant and is diagonalized by a
discrete Fourier transform.  Its real normal modes form $D_n$-sectors, whereas
the cyclic subgroup $C_n$ resolves each non-self-conjugate sector into two
characters with opposite phases.

This distinction separates the spectral and equivariant parts of the problem.
Commensurability among the modes active on a particular solution closes that
solution.  Commensurability of \emph{all} independent internal frequencies
makes the full center-of-mass-reduced quadratic Hamiltonian maximally
superintegrable
\cite{Evans1990Superintegrability,miller2013superintegrability,
Perelomov1990IntegrableSystems,tempesta2004superintegrability}. Neither periodicity nor maximal superintegrability determines how an orbit
transforms under particle relabeling. In particular, maximal
superintegrability guarantees closure of the bounded trajectories but does
not imply choreography. Particular integrals provide a complementary
mechanism by identifying invariant submanifolds that may support
choreographic motion~\cite{turbiner2013particular,escobar2024particular}.
Such a reduction is not sufficient by itself: the restricted dynamics must
also satisfy the character-wise phase conditions associated with cyclic
relabeling. In the present quadratic systems, these conditions are resolved
exactly under the time shift \(T/n\).

Our main results are as follows.  First, we give a Fourier criterion equivalent
to full $C_n$-equivariance and specialize it to harmonic and one-sided
traveling-wave solutions.  Second, for a fixed commensurate spectrum we
characterize the full linear subspace of equivariant initial data and compute
its codimension.  This replaces a qualitative appeal to phase locking by an
explicit condition on the positive- and negative-frequency amplitudes.  Third,
we analyze $n=4,5,6$.  The first two cases have only two inequivalent internal
branches.  The case $n=6$ has two doublets and a Nyquist singlet and therefore
distinguishes, for the first time, nondegenerate three-branch resonance from an
additional degeneracy between inequivalent characters.  The phase-matched
$1{:}2{:}3$ spectrum and the incompatible $1{:}2{:}2$ spectrum provide the
central comparison.

The present approach complements both the global classification of
choreographic symmetry groups \cite{MontaldiSteckles2013} and inverse-design
constructions in which a target curve is prescribed and the interactions are
recovered \cite{fernandez2025nbody}.  Our question is instead spectral: for a
fixed $D_n$-invariant quadratic Hamiltonian, which resonant initial data realize
the cyclic space--time symmetry?  It also extends the four-body analysis of
superintegrability and fragmentation in \cite{escobar2025four} to a general
character-level criterion and to the first three-branch case.

Section~\ref{sec:fourier} develops the normal-mode and Fourier framework.
Section~\ref{sec:initial-subspace} characterizes the equivariant initial-data
subspace.  Sections~\ref{sec:low-n} and~\ref{sec:n6} treat the low-$n$ examples,
and Section~\ref{sec:fragmentation} discusses multi-trace motion and the
large-$n$ structure.

\section{Dihedral harmonic networks and Fourier characters}
\label{sec:fourier}

\subsection{Hamiltonian and normal-mode spectrum}

Consider $n$ identical particles of mass $\mu$ in the plane, with positions
$\bm r_i\in\RR^2$ and conjugate momenta $\bm p_i$.  The Hamiltonian is
\begin{equation}
\mathcal H_n\ = \ 
\frac{1}{2\,\mu}\sum_{i=1}^n\bm p_i^2
\ + \ \frac{\mu\,\omega^2}{2}
\sum_{k=1}^{\lfloor n/2\rfloor}\kappa_k^{(n)}
\sum_{i=1}^n\lVert\bm r_i-\bm r_{i+k}\rVert^2,
\label{Hng}
\end{equation}
where indices are understood modulo $n$, $\omega>0$, and the couplings are
real.  Equation~\eqref{Hng} uses the opposite-pair double-count convention:
when $n$ is even, every pair in the $k=n/2$ term occurs twice.  If opposite
pairs are instead written once, then
\begin{equation}
\kappa_{n/2}^{(n)}\big|_{\rm single}
=2\kappa_{n/2}^{(n)}\big|_{\rm double}.
\label{eq:normalization}
\end{equation}

The center of mass decouples.  We impose
\[
\sum_{i=1}^n\bm r_i=0,
\qquad
\sum_{i=1}^n\bm p_i=0,
\]
identify the physical plane with $\CC$, and write $r_i=x_i+iy_i$.  With
$\zeta=e^{2\pi i/n}$, define
\begin{equation}
U_m=\frac{1}{\sqrt n}\sum_{j=1}^n r_j\zeta^{-m(j-1)},
\qquad
r_j=\frac{1}{\sqrt n}\sum_{m=0}^{n-1}U_m\zeta^{m(j-1)}.
\label{eq:DFT}
\end{equation}
The planar coefficients $U_m\in\CC$ are independent complex coordinates; in
particular, the scalar relation $U_{n-m}=\overline{U_m}$ need not hold.  The
corresponding real Cartesian-vector Fourier coefficients do satisfy the usual
componentwise reality relation.  The Supplementary Material gives the precise
dictionary between these two descriptions.

\begin{proposition}[Normal-mode spectrum]\label{prop:normal-modes}
In the center-of-mass frame, the label-Fourier coefficients satisfy
\begin{equation}
\ddot U_m+\Omega_m^2U_m=0,
\qquad
\Omega_m^2=4\omega^2
\sum_{k=1}^{\lfloor n/2\rfloor}
\kappa_k^{(n)}\sin^2\!\left(\frac{\pi mk}{n}\right),
\label{eq:general-spectrum}
\end{equation}
for $m=1,\ldots,n-1$.  Moreover,
$\Omega_m=\Omega_{n-m}$.  Hence $(m,n-m)$ forms one real cosine--sine
$D_n$-sector; for even $n$, $m=n/2$ is the self-conjugate Nyquist sector.
The bounded stable region is determined by $\Omega_m^2>0$ for every nonzero
character.
\end{proposition}

\begin{proof}
For the Fourier vector $v^{(m)}_j=n^{-1/2}\zeta^{m(j-1)}$, the $k$th bond
difference is multiplied by $1-\zeta^{mk}$.  The corresponding stiffness
eigenvalue is therefore
\[
\omega^2\sum_k\kappa_k^{(n)}
\lvert1-\zeta^{mk}\rvert^2
=4\omega^2\sum_k\kappa_k^{(n)}
\sin^2\!\left(\frac{\pi mk}{n}\right).
\]
It is unchanged under $m\mapsto n-m$.  The remaining statements follow from
the standard real decomposition of a symmetric circulant matrix.
\end{proof}

There are $d=2(n-1)$ relative degrees of freedom.  In the stable region the
system is Liouville integrable.  If every independent branch frequency is
rationally commensurate, the resonant oscillator is maximally
superintegrable and admits $2d-1$ functionally independent first integrals on
a dense open set \cite{Evans1990Superintegrability,miller2013superintegrability}.
If only a proper subset of branches is commensurate, this conclusion applies
to the associated invariant subsystem, not to the full Hamiltonian.

\subsection{Choreography and character-wise phase matching}

A collision-free $T$-periodic configuration
$R(t)=(r_1(t),\ldots,r_n(t))\in\CC^n$ is a single-trace choreography if
\begin{equation}
r_j(t)=r_1\!\left(t+\frac{(j-1)T}{n}\right),
\qquad j=1,\ldots,n.
\label{eq:choreo-def}
\end{equation}
No geometric $n$-fold symmetry of the generating curve is assumed.  Let
\[
c(z_1,\ldots,z_n)=(z_2,\ldots,z_n,z_1).
\]
The configuration is fully $C_n$-equivariant when
\begin{equation}
cR(t)=R\!\left(t+\frac{T}{n}\right).
\label{eq:configuration-equivariance}
\end{equation}

\begin{theorem}[Fourier criterion for $C_n$-equivariance]
\label{thm:phase-matching}
Let $R(t)$ be a $T$-periodic solution and let $U_m(t)$ be its label-Fourier
coefficients.  Then \eqref{eq:configuration-equivariance} holds if and only if
\begin{equation}
U_m\!\left(t+\frac{T}{n}\right)=\zeta^mU_m(t),
\qquad m=0,\ldots,n-1.
\label{eq:character-equivariance}
\end{equation}
In the center-of-mass frame the condition for $m=0$ is automatic.
\end{theorem}

\begin{proof}
Writing $R(t)=\sum_mU_m(t)v^{(m)}$ and using
$cv^{(m)}=\zeta^mv^{(m)}$ gives
\[
cR(t)=\sum_m\zeta^mU_m(t)v^{(m)},
\qquad
R\!\left(t+\frac{T}{n}\right)
=\sum_mU_m\!\left(t+\frac{T}{n}\right)v^{(m)}.
\]
The Fourier vectors form a basis, so the two configurations agree precisely
when their coefficients agree character by character.
\end{proof}

The sign of $\zeta^m$ follows from the convention chosen for $c$; the opposite
cyclic convention replaces it by $\zeta^{-m}$.  The harmonic form of the
criterion is particularly useful.

\begin{corollary}[Harmonic phase matching]
\label{cor:harmonic}
Suppose
\begin{equation}
U_m(t)=A_m e^{i\Omega_mt}+\widetilde A_m e^{-i\Omega_mt}.
\label{eq:harmonic-decomposition}
\end{equation}
Then character $m$ satisfies \eqref{eq:character-equivariance} if and only if
\begin{equation}
\bigl(e^{i\Omega_mT/n}-\zeta^m\bigr)A_m=0,
\qquad
\bigl(e^{-i\Omega_mT/n}-\zeta^m\bigr)\widetilde A_m=0.
\label{eq:harmonic-matching}
\end{equation}
\end{corollary}

\begin{proof}
Substitute \eqref{eq:harmonic-decomposition} into
\eqref{eq:character-equivariance} and equate the independent functions
$e^{i\Omega_mt}$ and $e^{-i\Omega_mt}$.
\end{proof}

For a one-sided traveling wave
\begin{equation}
U_m(t)=B_me^{i\Omega_mt},
\qquad B_m\neq0,
\label{eq:one-sided}
\end{equation}
the condition reduces to
\begin{equation}
e^{i\Omega_mT/n}=\zeta^m.
\label{eq:phase-matching}
\end{equation}
If $T=2\pi/\Omega_0$ and $\Omega_m=k_m\Omega_0$ with $k_m\in\ZZ$, then
\begin{equation}
k_m\equiv m\pmod n.
\label{eq:congruence}
\end{equation}
Constant phases multiplying $B_m$ do not affect this congruence; incompatible
temporal orientations or characters must be suppressed rather than repaired by
a relative phase choice.

\begin{corollary}\label{cor:equiv-choreo}
A collision-free periodic configuration satisfying
\eqref{eq:configuration-equivariance} is a single-trace choreography.
\end{corollary}

\begin{proof}
Componentwise, \eqref{eq:configuration-equivariance} gives
$r_{i+1}(t)=r_i(t+T/n)$.  Iteration yields \eqref{eq:choreo-def}.
\end{proof}

\section{The equivariant initial-data subspace}
\label{sec:initial-subspace}

The preceding criterion determines not only whether a selected traveling wave
is choreographic, but the complete set of equivariant initial data for a fixed
commensurate spectrum.  Let $\cP_T$ denote the real vector space of
center-of-mass-free solutions that are periodic with a chosen common period
$T$.  For each periodic character, write
\[
k_m=\frac{\Omega_mT}{2\pi}\in\ZZ_{>0}
\]
and define
\begin{equation}
\chi_m^+=
\begin{cases}
1,& k_m\equiv m\pmod n,\\
0,& \text{otherwise},
\end{cases}
\qquad
\chi_m^-=
\begin{cases}
1,& -k_m\equiv m\pmod n,\\
0,& \text{otherwise}.
\end{cases}
\label{eq:matching-indicators}
\end{equation}

\begin{proposition}[Equivariant initial-data subspace]
\label{prop:equivariant-subspace}
The $C_n$-equivariant solutions of period $T$ form the linear subspace
$\cE_T\subset\cP_T$ defined, in the decomposition
\eqref{eq:harmonic-decomposition}, by
\begin{equation}
A_m=0\quad\text{if }\chi_m^+=0,
\qquad
\widetilde A_m=0\quad\text{if }\chi_m^-=0.
\label{eq:equivariant-amplitudes}
\end{equation}
If $\cI_T\subset\{1,\ldots,n-1\}$ denotes the periodic characters included in
$\cP_T$, then
\begin{align}
\dim_{\RR}\cE_T
&=2\sum_{m\in\cI_T}(\chi_m^++\chi_m^-),
\label{eq:equivariant-dimension}\\
\operatorname{codim}_{\cP_T}\cE_T
&=4\lvert\cI_T\rvert
-2\sum_{m\in\cI_T}(\chi_m^++\chi_m^-).
\label{eq:equivariant-codimension}
\end{align}
Here each allowed complex amplitude contributes two real dimensions.
\end{proposition}

\begin{proof}
Because $\Omega_mT=2\pi k_m$, the two temporal multipliers under $T/n$ are
$e^{\pm2\pi ik_m/n}$.  Corollary~\ref{cor:harmonic} allows $A_m$ precisely
when $k_m\equiv m\pmod n$ and allows $\widetilde A_m$ precisely when
$-k_m\equiv m\pmod n$.  All other amplitudes vanish.  The constraints are
homogeneous and diagonal in Fourier space, so they define a linear subspace.
Counting the allowed complex amplitudes gives
\eqref{eq:equivariant-dimension} and
\eqref{eq:equivariant-codimension}.
\end{proof}

For a non-Nyquist character at most one temporal orientation can satisfy the
condition.  For the Nyquist character $m=n/2$, both orientations are allowed
when $k_m\equiv n/2\pmod n$, because $\zeta^{n/2}=-1$ is its own inverse.
If $\cE_T$ is proper, it has empty interior and zero Lebesgue measure in
$\cP_T$.  Thus generic periodic initial data are not $C_n$-equivariant even
though every orbit in $\cP_T$ closes.  This statement is about initial data on
the exact resonance surface; away from that surface, the motion is generally
quasiperiodic.

This result also sharpens the relation with maximal superintegrability.  A
commensurate full spectrum makes every bounded relative orbit periodic, but
only the subspace $\cE_T$ realizes the choreography generator.  Spectral
degeneracy can enlarge symmetry algebras and the space of first integrals, but
it cannot overcome a mismatch between the temporal multiplier and the
$C_n$-character.

\section{Four- and five-body cases}
\label{sec:low-n}

The cases $n=4$ and $n=5$ each have two independent internal frequencies and
provide compact checks of the general criterion.  The explicit real normal
coordinates, trajectories, and plotted initial data are collected in the
Supplementary Material.

\subsection{Four bodies}

Using a single count for the opposite pairs, the $D_4$-invariant potential is
\[
\begin{aligned}
V_4
={}&
\frac{\mu\omega^2}{2}
\Bigl[
\kappa_1^{(4)}
\left(
r_{12}^2+r_{23}^2+r_{34}^2+r_{14}^2
\right)
+\kappa_2^{(4)}
\left(
r_{13}^2+r_{24}^2
\right)
\Bigr].
\end{aligned}
\]
The characters $(1,3)$ form a doublet with frequency $\Omega_{\rm F}$ and
$m=2$ is the Nyquist character with frequency $\Omega_{\rm N}$, where
\begin{equation}
\Omega_{\rm F}
=\omega\sqrt{2(\kappa_1^{(4)}+\kappa_2^{(4)})},
\qquad
\Omega_{\rm N}=2\omega\sqrt{\kappa_1^{(4)}}.
\label{eq:n4-frequencies}
\end{equation}
Stability requires $\kappa_1^{(4)}>0$ and
$\kappa_1^{(4)}+\kappa_2^{(4)}>0$.  The full relative Hamiltonian is maximally
superintegrable when $p=\Omega_{\rm N}/\Omega_{\rm F}\in\QQ$.  In the
one-sided realization with active characters $m=1,2$ and
$T=2\pi/\Omega_{\rm F}$, phase matching requires
\[
p\equiv2\pmod4.
\]
The primitive choice $p=2$ gives
\begin{equation}
\kappa_2^{(4)}=-\frac12\kappa_1^{(4)}.
\label{eq:n4-primitive}
\end{equation}
For $(\kappa_1^{(4)},\kappa_2^{(4)})=(1,-1/2)$, compatible data give the
single-trace choreography in Fig.~\ref{fig:n4-examples}(b); generic periodic
data can instead produce the two-dimer pattern in
Fig.~\ref{fig:n4-examples}(a).

\begin{figure}[htbp]
\centering
\begin{minipage}[t]{0.48\textwidth}
\centering
\includegraphics[width=0.92\linewidth]{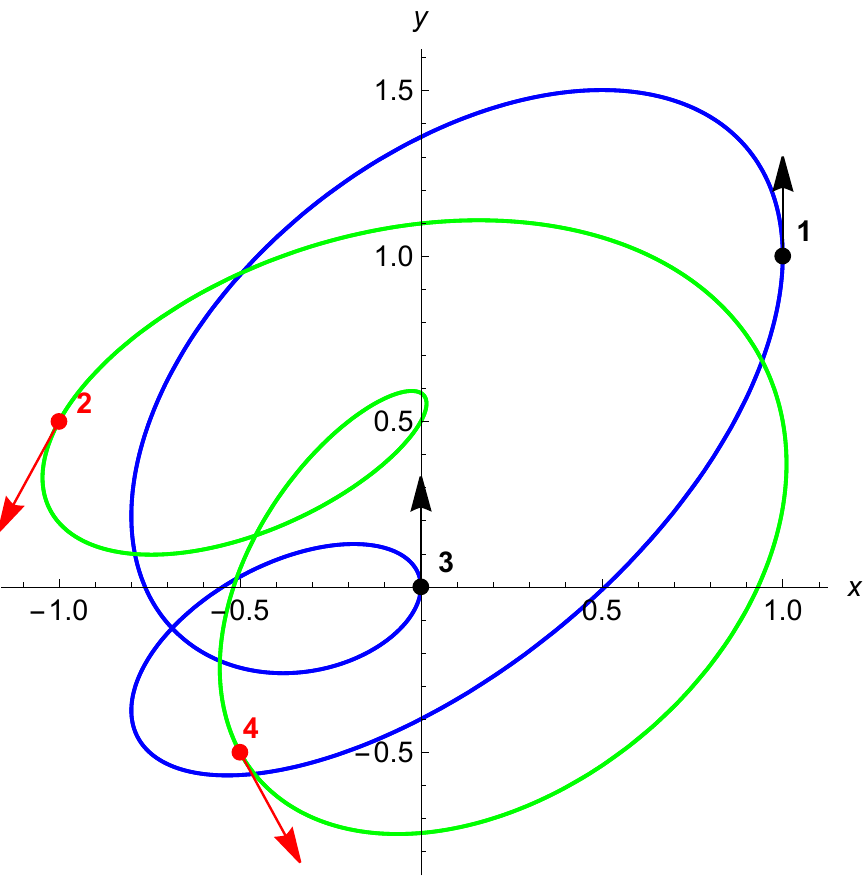}\\[-2pt]
\small (a) $(2+2)$ fragmentation
\end{minipage}\hfill
\begin{minipage}[t]{0.48\textwidth}
\centering
\includegraphics[width=0.84\linewidth]{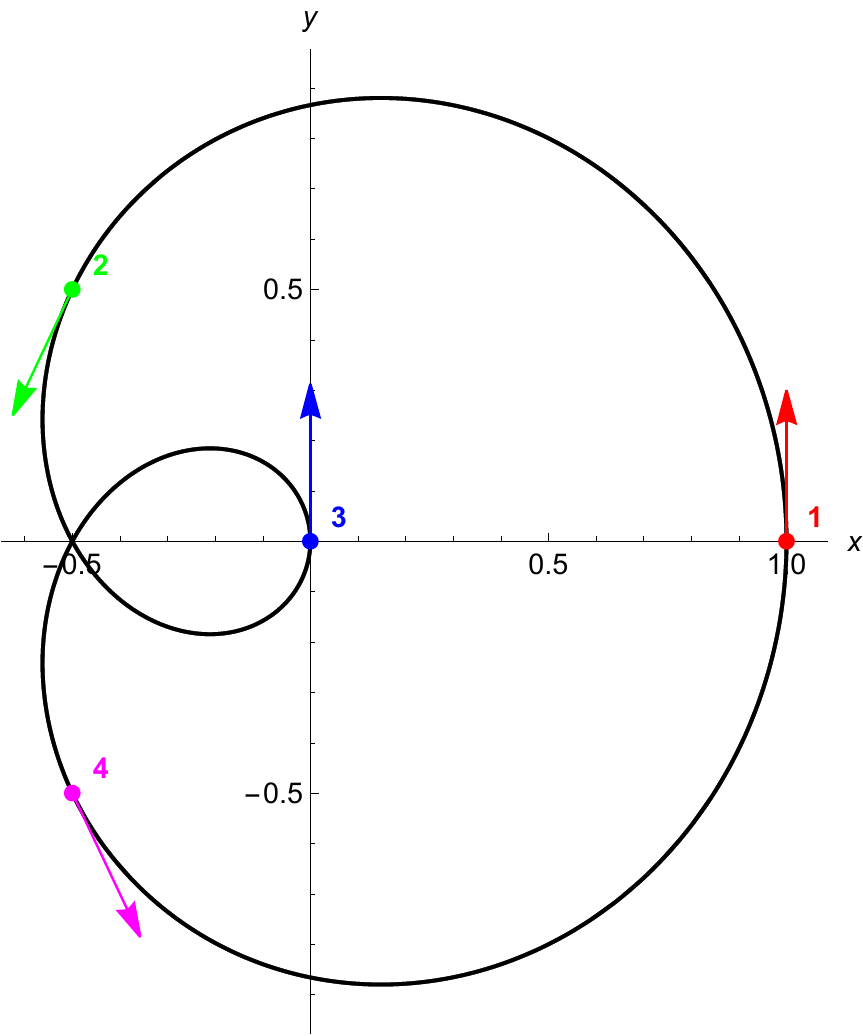}\\[-2pt]
\small (b) $C_4$-equivariant choreography
\end{minipage}
\caption{Four-body motion at the $1{:}2$ resonance with
$\mu=\omega=1$ and $(\kappa_1^{(4)},\kappa_2^{(4)})=(1,-1/2)$.
In (a), the particles form two synchronized dimers and the full configuration
is only $C_2$-equivariant.  In (b), all particles traverse the same
lima\c{c}on-type trace with delay $T/4$.  Initial data and collision checks are
given in the Supplementary Material.}
\label{fig:n4-examples}
\end{figure}

\subsection{Five bodies}

For
\[
\begin{aligned}
V_5
={}&
\frac{\mu\omega^2}{2}
\Bigl[
\kappa_1^{(5)}
\left(
r_{12}^2+r_{23}^2+r_{34}^2+r_{45}^2+r_{15}^2
\right)
+\kappa_2^{(5)}
\left(
r_{13}^2+r_{24}^2+r_{35}^2+r_{14}^2+r_{25}^2
\right)
\Bigr] \ ,
\end{aligned}
\]
the characters $(1,4)$ and $(2,3)$ form two real $D_5$-doublets.  Their
frequencies are $\Omega_1=\omega\sqrt{\lambda_-}$ and
$\Omega_2=\omega\sqrt{\lambda_+}$, with
\begin{equation}
\lambda_\pm=\frac12\left[
5(\kappa_1^{(5)}+\kappa_2^{(5)})
\pm\sqrt5(\kappa_1^{(5)}-\kappa_2^{(5)})\right].
\label{eq:n5-frequencies}
\end{equation}
The stable region has $\lambda_\pm>0$.  Maximal superintegrability requires
$p=\Omega_2/\Omega_1\in\QQ$.  Selecting the one-sided characters $m=1,2$
and the primitive period $T=2\pi/\Omega_1$ gives
\[
p\equiv2\pmod5.
\]
The lowest positive solution is again $p=2$.  One convenient coupling choice is
\begin{equation}
\kappa_1^{(5)}=\frac12\left(\frac3{\sqrt5}+1\right),
\qquad
\kappa_2^{(5)}=-\frac12\left(\frac3{\sqrt5}-1\right),
\label{eq:n5-couplings}
\end{equation}
for which $\Omega_2=2\Omega_1$.  Figure~\ref{fig:n5-examples} contrasts a
periodic $(2+2+1)$ fragmentation with the phase-matched five-body
choreography.

\begin{figure}[htbp]
\centering
\begin{minipage}[t]{0.48\textwidth}
\centering
\includegraphics[width=0.94\linewidth]{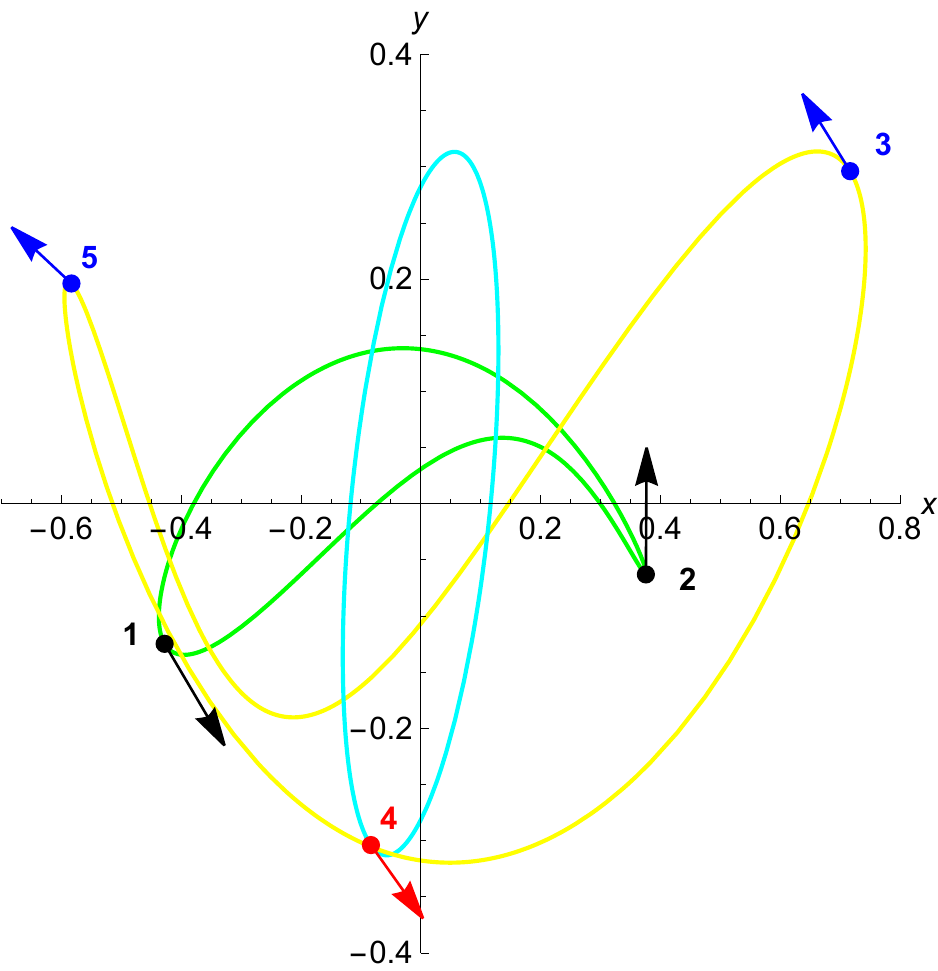}\\[-2pt]
\small (a) $(2+2+1)$ fragmentation
\end{minipage}\hfill
\begin{minipage}[t]{0.48\textwidth}
\centering
\includegraphics[width=0.94\linewidth]{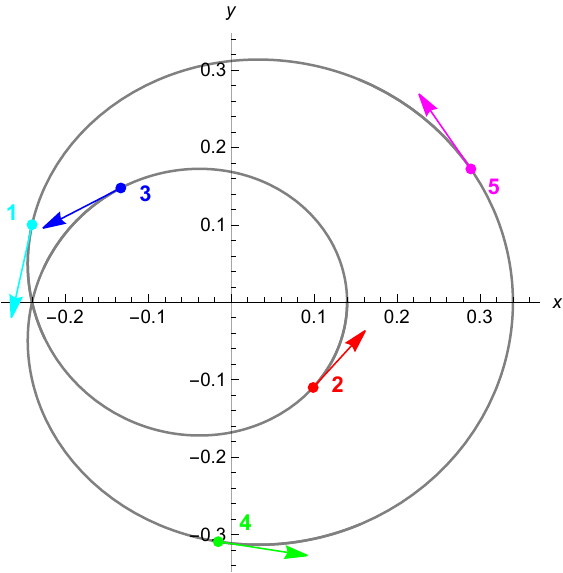}\\[-2pt]
\small (b) $C_5$-equivariant choreography
\end{minipage}
\caption{Five-body motion at the $1{:}2$ resonance for the couplings
\eqref{eq:n5-couplings} and $\mu=\omega=1$.  The periodic motion in (a) is not
globally $C_5$-equivariant.  In (b), the five particles traverse one
lima\c{c}on-type curve with delay $T/5$.  Initial data and collision checks are
given in the Supplementary Material.}
\label{fig:n5-examples}
\end{figure}

\section{Six bodies: resonance versus degeneracy}
\label{sec:n6}

The six-body problem has three inequivalent internal branches: the doublets
$(1,5)$ and $(2,4)$ and the Nyquist character $m=3$.  It is therefore the
first case with two independent resonance ratios.

\subsection{Spectrum and maximal superintegrability}

We use a single count for the opposite pairs and write
\begin{equation}
\begin{aligned}
\mathcal H_6
={}&
\frac{1}{2\mu}
\left(
\bm p_1^2+\bm p_2^2+\bm p_3^2
+\bm p_4^2+\bm p_5^2+\bm p_6^2
\right)
+\frac{\mu\omega^2}{2}
\Bigl[
\kappa_1^{(6)}
\left(
r_{12}^2+r_{23}^2+r_{34}^2
+r_{45}^2+r_{56}^2+r_{16}^2
\right)
\\
&
+\ \kappa_2^{(6)}
\left(
r_{13}^2+r_{24}^2+r_{35}^2
+r_{46}^2+r_{15}^2+r_{26}^2
\right)
\ + \ \kappa_3^{(6)}
\left(
r_{14}^2+r_{25}^2+r_{36}^2
\right)
\Bigr].
\end{aligned}
\label{eq:H6}
\end{equation}
The three stiffness eigenvalues are
\begin{equation}
\lambda_1=\kappa_1^{(6)}+3\kappa_2^{(6)}+2\kappa_3^{(6)},
\quad
\lambda_2=3\kappa_1^{(6)}+3\kappa_2^{(6)},
\quad
\lambda_3=4\kappa_1^{(6)}+2\kappa_3^{(6)},
\label{eq:n6-stiffness}
\end{equation}
and $\Omega_\ell=\omega\sqrt{\lambda_\ell}$.  Stability requires
$\lambda_\ell>0$.  The ten-degree-of-freedom relative Hamiltonian is maximally
superintegrable when
\begin{equation}
\frac{\Omega_2}{\Omega_1}\in\QQ,
\qquad
\frac{\Omega_3}{\Omega_1}\in\QQ,
\label{eq:n6-superintegrability}
\end{equation}
in which case the resonant-oscillator construction gives $19$ functionally
independent first integrals on a dense open set.

The linear map
$(\kappa_1^{(6)},\kappa_2^{(6)},\kappa_3^{(6)})\mapsto
(\lambda_1,\lambda_2,\lambda_3)$ in \eqref{eq:n6-stiffness} has determinant
$-36$ and is therefore invertible.  Consequently, prescribed ratios
$\Omega_2/\Omega_1=p$ and $\Omega_3/\Omega_1=q$ determine a unique ray in
coupling space: the overall scale fixes $\Omega_1$, while the direction is
fixed by $(\lambda_1,\lambda_2,\lambda_3)\propto(1,p^2,q^2)$.

\subsection{The nondegenerate \texorpdfstring{$1{:}2{:}3$}{1:2:3} resonance}

Let $T=2\pi/\Omega_1$ and activate the one-sided characters $m=1,2,3$.
Equation~\eqref{eq:phase-matching} gives
\[
e^{i\Omega_1T/6}=\zeta,
\qquad
e^{i\Omega_2T/6}=\zeta^2,
\qquad
e^{i\Omega_3T/6}=\zeta^3,
\]
whose primitive solution is
\begin{equation}
\Omega_1:\Omega_2:\Omega_3=1:2:3.
\label{eq:n6-123}
\end{equation}
For example,
\begin{equation}
(\kappa_1^{(6)},\kappa_2^{(6)},\kappa_3^{(6)})
=\left(2,-\frac23,\frac12\right)
\label{eq:n6-123-couplings}
\end{equation}
gives $(\lambda_1,\lambda_2,\lambda_3)=(1,4,9)$.

The phase condition also gives existence of collision-free, genuinely
three-branch choreographies.  With only $B_1\neq0$, the inverse Fourier
transform describes a rotating regular hexagon and has a strictly positive
minimum interparticle separation.  Turning on sufficiently small nonzero
$B_2$ and $B_3$ preserves that separation by continuity on the compact period
interval, while all three characters remain phase matched.  Thus the
$1{:}2{:}3$ Hamiltonian supports an open set, within its equivariant subspace,
of collision-free three-branch choreographies.  Proposition
\ref{prop:equivariant-subspace} nevertheless shows that this equivariant
subspace is nongeneric in the full periodic initial-data space.

An explicit member of this set is obtained for $\mu=\omega=1$ and the
couplings \eqref{eq:n6-123-couplings} by taking
\begin{equation}
U_1(t)=\sqrt6\,e^{it},\qquad
U_2(t)=\frac{\sqrt6}{3}e^{2it},\qquad
U_3(t)=\frac{\sqrt6}{6}e^{3it},
\qquad U_4(t)=U_5(t)=0.
\label{eq:n6-explicit-123}
\end{equation}
Writing $\theta_j=2\pi(j-1)/6$, the particle trajectories are
\begin{equation}
r_j(t)=\gamma(t+\theta_j),
\qquad
\gamma(s)=e^{is}+\frac13e^{2is}+\frac16e^{3is}.
\label{eq:n6-generating-curve}
\end{equation}
All three internal branches are nonzero.  Moreover, for every distinct pair
the phase difference is $\Delta_\ell=2\pi\ell/6$, with
$\ell=1,\ldots,5$, and the triangle inequality gives
\begin{equation}
\left|\gamma(s+\Delta_\ell)-\gamma(s)\right|
\geq
2\left|\sin\frac{\Delta_\ell}{2}\right|
-\frac23\left|\sin\Delta_\ell\right|
-\frac13\left|\sin\frac{3\Delta_\ell}{2}\right|
\geq \frac{2-\sqrt3}{3}>0.
\label{eq:n6-collision-bound}
\end{equation}
Hence the choreography is analytically collision free.  Direct minimization
over one period gives $d_{\min}=0.577350$; the configuration at $t=0$ is shown
in Fig.~\ref{fig:n6-full-123}.

\begin{figure}[htbp]
\centering
\includegraphics[width=0.62\linewidth]{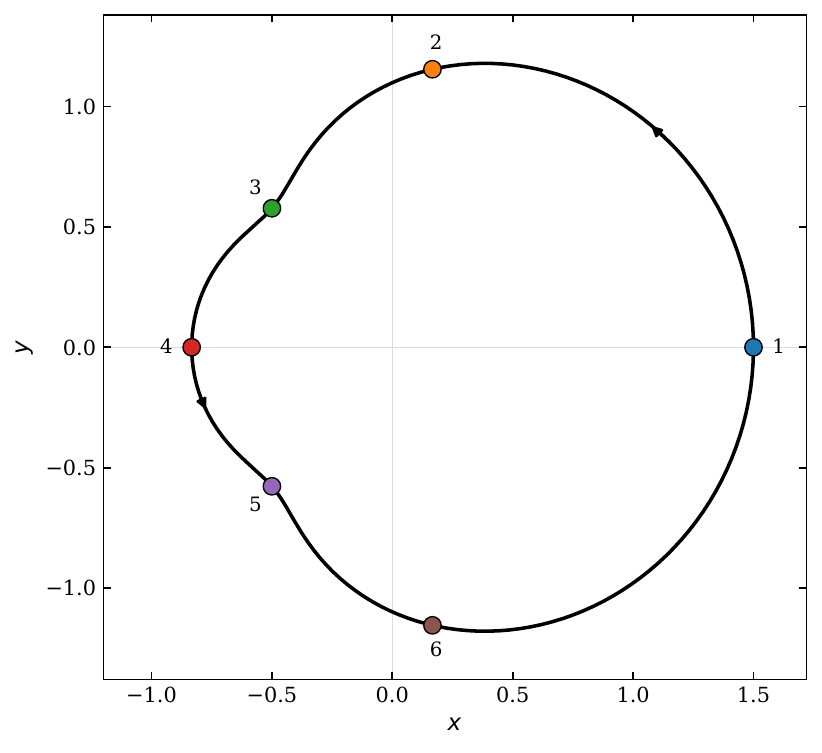}
\caption{Explicit full three-branch $C_6$ choreography at the
$1{:}2{:}3$ resonance.  The black curve is the generating trace
$\gamma$ in \eqref{eq:n6-generating-curve}; colored markers show the six
particles at $t=0$.  Every particle traverses the same curve with delay $T/6$,
and all three Fourier amplitudes in \eqref{eq:n6-explicit-123} are nonzero.}
\label{fig:n6-full-123}
\end{figure}

\subsection{The exactly degenerate \texorpdfstring{$1{:}2{:}2$}{1:2:2} spectrum}

Now impose
\begin{equation}
\Omega_1:\Omega_2:\Omega_3=1:2:2,
\qquad \Omega_2=\Omega_3.
\label{eq:n6-122}
\end{equation}
Equations~\eqref{eq:n6-stiffness} give the ray
\begin{equation}
(\kappa_1^{(6)},\kappa_2^{(6)},\kappa_3^{(6)})=(7a,a,-2a),
\qquad a>0.
\label{eq:n6-122-ray}
\end{equation}
Under $T/6$, the degenerate branches acquire the common temporal multiplier
\[
e^{i\Omega_2T/6}=e^{i\Omega_3T/6}=\zeta^2.
\]
This matches the $m=2$ character but not the Nyquist character $m=3$, which
requires $\zeta^3=-1$.  Multiplying either amplitude by a constant phase does
not change its temporal multiplier.  Hence simultaneous nonzero one-sided
$m=2$ and $m=3$ excitation is incompatible with full $C_6$-equivariance.
Exact spectral degeneracy does not merge the two characters.

A choreography remains possible after restricting the initial data to a
compatible invariant subspace.  Setting $U_3=0$ leaves the phase-matched
$m=1,2$ content with active ratio $1{:}2$.  Taking $a=1/2$ and
$\omega=1/\sqrt3$ gives
\[
(\kappa_1^{(6)},\kappa_2^{(6)},\kappa_3^{(6)})
=\left(\frac72,\frac12,-1\right),
\qquad
(\Omega_1,\Omega_2,\Omega_3)=(1,2,2).
\]

The choreography in Fig.~\ref{fig:n6-examples}(b) is therefore a reduced
$1{:}2$ realization embedded in the degenerate Hamiltonian; it is not produced
by the degeneracy itself.

\begin{figure}[htbp]
\centering
\begin{minipage}[t]{0.48\textwidth}
\centering
\includegraphics[width=0.90\linewidth]{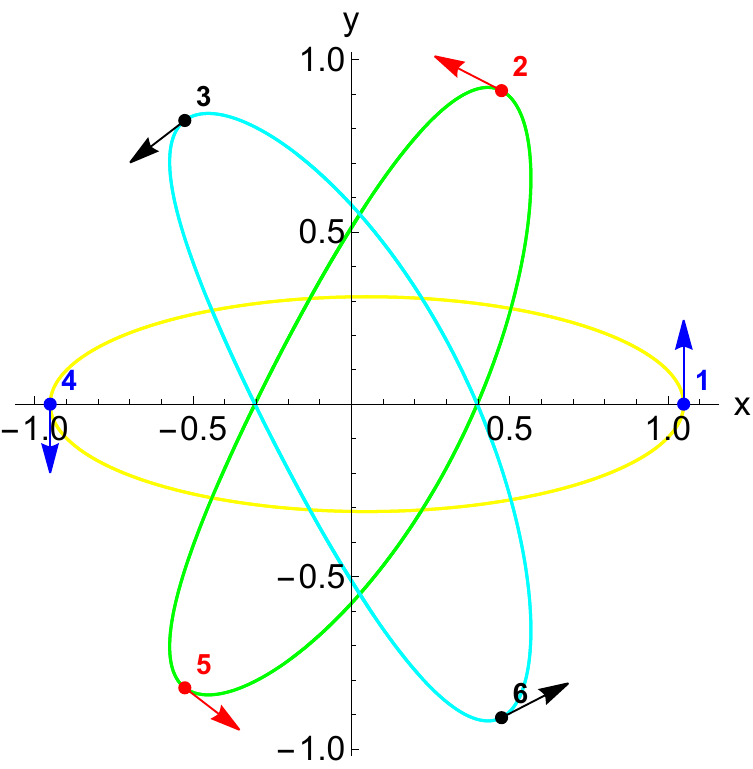}\\[-2pt]
\small (a) $(2+2+2)$ fragmentation
\end{minipage}\hfill
\begin{minipage}[t]{0.48\textwidth}
\centering
\includegraphics[width=0.90\linewidth]{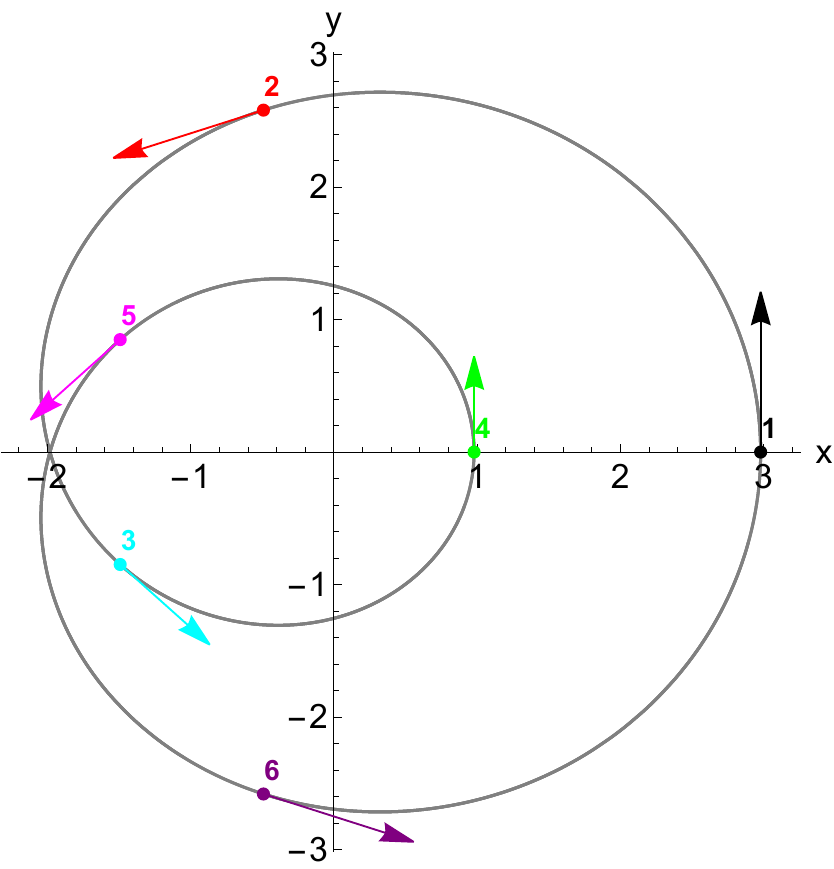}\\[-2pt]
\small (b) reduced $C_6$ choreography
\end{minipage}
\caption{Two six-body realizations.  Panel (a) uses
$\mu=\omega=1$ and the $1{:}2{:}3$ couplings
\eqref{eq:n6-123-couplings}; the particles form three synchronized dimers and
the full motion is not $C_6$-equivariant.  Panel (b) uses
$\mu=1$, $\omega=1/\sqrt3$, and the $1{:}2{:}2$ couplings
$(7/2,1/2,-1)$, but lies in the invariant subspace $U_3=0$; all particles
traverse one trace with delay $T/6$.  Initial data and collision checks are
given in the Supplementary Material.}
\label{fig:n6-examples}
\end{figure}

The three spectral mechanisms are summarized in
Table~\ref{tab:n6-summary}.  In particular, branch suppression is different
from phase matching across a full three-branch spectrum.

\begin{table}[htbp]
\centering
\small
\renewcommand{\arraystretch}{1.15}
\begin{tabular}{>{\raggedright\arraybackslash}p{0.23\textwidth}
                >{\centering\arraybackslash}p{0.18\textwidth}
                >{\raggedright\arraybackslash}p{0.51\textwidth}}
\toprule
Setting & Active spectrum & Consequence\\
\midrule
Nondegenerate resonance & $1{:}2{:}3$ &
The one-sided $m=1,2,3$ components are simultaneously phase matched; a
collision-free full three-branch choreography exists.\\
Exact degeneracy & $1{:}2{:}2$ &
The $m=2$ and $m=3$ components have the same temporal phase but require
different character phases; generic simultaneous excitation is obstructed.\\
Invariant-subspace reduction & $1{:}2$ &
Suppressing $U_3$ leaves compatible $m=1,2$ content and yields a reduced
six-body choreography.\\
\bottomrule
\end{tabular}
\caption{The three distinct six-body mechanisms.}
\label{tab:n6-summary}
\end{table}

\section{Multi-trace motion and fragmentation}
\label{sec:fragmentation}

Frequency commensurability, full equivariance, choreography, and fragmentation
are distinct notions.  Commensurability closes the active motion.
Equation~\eqref{eq:character-equivariance} selects the configurations with the
global cyclic space--time symmetry, and Corollary~\ref{cor:equiv-choreo} turns
that symmetry into a single-trace choreography in the collision-free class.
Failure of phase matching rules out such a choreography, but does not by
itself determine the geometry of the remaining periodic motion.

We call a periodic motion \emph{choreographically fragmented} only when the
particles split into two or more synchronized subsets, each subset traversing
its own closed curve with a uniform internal delay.  Thus fragmentation is a
structured possibility inside the non-equivariant periodic set, not a synonym
for every phase mismatch.  The displayed examples realize $(2+2)$,
$(2+2+1)$, and $(2+2+2)$ patterns.  Their residual cyclic symmetries act on the
submotions rather than on the complete $n$-particle configuration.

For $n>6$, the number of independent real sectors is $\lfloor n/2\rfloor$.
Consequently, more resonance integers and more pairs of temporal orientations
enter Proposition~\ref{prop:equivariant-subspace}.  The decision procedure is
nevertheless unchanged:
\begin{enumerate}
\item determine the stable branch frequencies from
      \eqref{eq:general-spectrum};
\item choose a common period for the commensurate active branches;
\item apply the congruences $\pm k_m\equiv m\pmod n$ to every Fourier
      amplitude;
\item reconstruct the configuration and test collision-freeness.
\end{enumerate}
The resulting equivariant space is explicit and linear, whereas the geometry
of its non-equivariant complement may range from unstructured multi-trace
motion to synchronized fragmentation.  The low-$n$ examples are therefore
illustrations of a general character-level obstruction, not an exhaustive
classification of fragmentations.

\section{Conclusions}\label{sec:conclusions}

Quadratic $D_n$-invariant many-body Hamiltonians provide an exact setting in
which spectral closure and choreographic symmetry can be separated.  The
normal frequencies are determined by the circulant spectrum
\eqref{eq:general-spectrum}.  Rational commensurability of all independent
branches makes the full relative Hamiltonian maximally superintegrable, while
commensurability of only the active branches closes the corresponding orbit or
invariant subsystem.  Choreography requires the additional character-wise
condition \eqref{eq:character-equivariance}.

For a fixed commensurate spectrum, Proposition
\ref{prop:equivariant-subspace} identifies the complete equivariant subspace of
initial data and its codimension.  It makes precise why resonant periodicity is
generic within a maximally superintegrable oscillator but choreography is not:
incompatible positive- or negative-frequency amplitudes must vanish.  Constant
relative phases cannot remove the obstruction.

The six-body system exhibits the essential new mechanism.  Its
$1{:}2{:}3$ spectrum phase matches the three characters $m=1,2,3$ and supports
collision-free three-branch choreographies.  In the $1{:}2{:}2$ spectrum, the
additional degeneracy $\Omega_2=\Omega_3$ is instead incompatible with
simultaneous $m=2$ and Nyquist $m=3$ excitation.  Suppressing the Nyquist
branch yields a valid reduced choreography, but the degeneracy itself is not
its dynamical origin.  Periodic data outside the equivariant subspace may
organize into synchronized multi-trace fragments.

The same framework applies to arbitrary $n$ and suggests two natural
extensions: persistence of the equivariant subspaces under
symmetry-preserving nonlinear perturbations, and a subgroup-based
classification of the possible fragmented configurations.

\subsection*{Declaration of competing interest}

The authors declare that they have no known competing financial interests or personal relationships that could have appeared to influence the work reported in this paper.

\subsection*{Data availability}

No external datasets were used in this work. The trajectories are generated from analytic formulas and initial conditions reported in the main text and Supplementary Material.

\bibliographystyle{unsrt}
\bibliography{references}

@incollection{ChencinerGerverMontgomerySimo2002,
  author    = {Chenciner, Alain and Gerver, Joseph and Montgomery, Richard and Sim{\'o}, Carles},
  title     = {Simple Choreographic Motions of {N} Bodies: A Preliminary Study},
  booktitle = {Geometry, Mechanics, and Dynamics},
  editor    = {Newton, Philip and Holmes, Philip and Weinstein, Alan},
  publisher = {Springer},
  address   = {New York, NY},
  year      = {2002},
  pages     = {287--308},
  doi       = {10.1007/0-387-21791-6_9}
}

@article{Yu2017,
  author  = {Yu, Guowei},
  title   = {Simple Choreographies of the Planar Newtonian {$N$}-Body Problem},
  journal = {Archive for Rational Mechanics and Analysis},
  volume  = {225},
  pages   = {901--935},
  year    = {2017},
  doi     = {10.1007/s00205-017-1116-1}
}

@article{moore1993braids,
  author    = {Moore, Christopher},
  title     = {Braids in classical dynamics},
  journal   = {Physical Review Letters},
  volume    = {70},
  number    = {24},
  pages     = {3675--3679},
  year      = {1993},
  doi       = {10.1103/PhysRevLett.70.3675}
}

@article{chenciner2000remarkable,
  author    = {Chenciner, Alain and Montgomery, Richard},
  title     = {A remarkable periodic solution of the three-body problem in the case of equal masses},
  journal   = {Annals of Mathematics},
  volume    = {152},
  number    = {3},
  pages     = {881--901},
  year      = {2000},
  doi       = {10.2307/2661357}
}

@incollection{simo2002dynamical,
  author    = {Sim{\'o}, Carles},
  title     = {Dynamical properties of the figure-eight solution of the three-body problem},
  booktitle = {Celestial Mechanics},
  editor    = {Celletti, Alessandra and Sim{\'o}, Carles},
  series    = {Contemporary Mathematics},
  volume    = {292},
  pages     = {209--228},
  publisher = {American Mathematical Society},
  address   = {Providence, RI},
  year      = {2002},
  isbn      = {978-0-8218-2881-7}
}

@article{ferrario2004relative,
  author    = {Ferrario, Davide and Terracini, Susanna},
  title     = {On the existence of collisionless equivariant minimizers for the classical $n$-body problem},
  journal   = {Inventiones Mathematicae},
  volume    = {155},
  number    = {2},
  pages     = {305--362},
  year      = {2004},
  doi       = {10.1007/s00222-003-0322-7}
}

@article{MontaldiSteckles2013,
  author  = {Montaldi, James and Steckles, Katrina},
  title   = {Classification of symmetry groups for planar {$n$}-body choreographies},
  journal = {Forum of Mathematics, Sigma},
  volume  = {1},
  pages   = {e5},
  year    = {2013},
  doi     = {10.1017/fms.2013.5},
  eprint  = {1305.0470},
  archivePrefix = {arXiv}
}

@article{fernandez2025fourbody,
  author    = {Fern{\'a}ndez-Guasti, Manuel},
  title     = {Analytic four-body lima{\c{c}}on choreography},
  journal   = {Celestial Mechanics and Dynamical Astronomy},
  volume    = {137},
  number    = {4},
  pages     = {4},
  year      = {2025},
  doi       = {10.1007/s10569-024-10235-x},
  url       = {https://doi.org/10.1007/s10569-024-10235-x}
}

@article{escobar2025four,
  author    = {Escobar-Ruiz, Adri{\'a}n M. and Fern{\'a}ndez-Guasti, Manuel},
  title     = {On the four-body lima{\c{c}}on choreography: maximal superintegrability and choreographic fragmentation},
  journal   = {Celestial Mechanics and Dynamical Astronomy},
  volume    = {137},
  number    = {24},
  pages     = {24},
  year      = {2025},
  doi       = {10.1007/s10569-025-10255-1}
}

@article{fernandez2025nbody,
  title = {N-body choreographies on a p-limaçon curve},
journal = {Journal of Differential Equations},
volume = {454},
pages = {113940},
year = {2026},
issn = {0022-0396},
doi = {10.1016/j.jde.2025.113940},
url = {https://www.sciencedirect.com/science/article/pii/S0022039625009672},
author = {Manuel Fernandez-Guasti and Toshiaki Fujiwara and Ernesto Pérez-Chavela and Shuqiang Zhu}
}

@article{miller2013superintegrability,
  author    = {Miller, Willard and Post, Sarah and Winternitz, Pavel},
  title     = {Classical and quantum superintegrability with applications},
  journal   = {Journal of Physics A: Mathematical and Theoretical},
  volume    = {46},
  pages     = {423001},
  year      = {2013},
  doi       = {10.1088/1751-8113/46/42/423001}
}

@book{tempesta2004superintegrability,
  editor    = {Tempesta, Piergiulio and Winternitz, Pavel and Harnad, Jonathan and Miller Jr., Willard and Pogosyan, George and Rodríguez, Miguel A.},
  title     = {Superintegrability in Classical and Quantum Systems},
  series    = {CRM Proceedings \& Lecture Notes},
  volume    = {37},
  publisher = {American Mathematical Society},
  address   = {Providence, RI},
  year      = {2004},
  isbn      = {978-0-8218-3329-2}
}

@article{escobar2024particular,
doi = {10.1088/1751-8121/ad2a1c},
url = {https://doi.org/10.1088/1751-8121/ad2a1c},
year = {2024},
month = {feb},
publisher = {IOP Publishing},
volume = {57},
number = {10},
pages = {105202},
author = {Escobar-Ruiz, A M and Azuaje, R},
title = {On particular integrability in classical mechanics},
journal = {Journal of Physics A: Mathematical and Theoretical}
}

@article{turbiner2013particular,
  
doi = {10.1088/1751-8113/46/2/025203},
url = {https://doi.org/10.1088/1751-8113/46/2/025203},
year = {2012},
month = {dec},
publisher = {IOP Publishing},
volume = {46},
number = {2},
pages = {025203},
author = {Turbiner, Alexander V},
title = {Particular integrability and (quasi)-exact-solvability},
journal = {Journal of Physics A: Mathematical and Theoretical}
}

@article{barutello2006double,
  author  = {Barutello, Vivina and Terracini, Susanna},
  title   = {Double choreographical solutions for $n$‐body type problems},
  journal = {Celestial Mechanics and Dynamical Astronomy},
  volume  = {95},
  number  = {1--4},
  pages   = {67--80},
  year    = {2006},
  doi     = {10.1007/s10569-006-9030-0}
}

@article{montgomery1998braid,
  author  = {Montgomery, Richard},
  title   = {The $N$-body problem, the braid group, and action-minimizing periodic solutions},
  journal = {Nonlinearity},
  volume  = {11},
  number  = {2},
  pages   = {363--376},
  year    = {1998},
  doi     = {10.1088/0951-7715/11/2/011}
}

@article{kapela2007symbolic,
 doi = {10.1088/0951-7715/20/5/010},
url = {https://doi.org/10.1088/0951-7715/20/5/010},
year = {2007},
month = {apr},
publisher = {},
volume = {20},
number = {5},
pages = {1241},
author = {Kapela, Tomasz and Simó, Carles},
title = {Computer assisted proofs for nonsymmetric planar choreographies and for stability of the Eight},
journal = {Nonlinearity}
}

@article{fujiwara2003lemniscate,
doi = {10.1088/0305-4470/36/11/310},
url = {https://doi.org/10.1088/0305-4470/36/11/310},
year = {2003},
month = {mar},
publisher = {},
volume = {36},
number = {11},
pages = {2791},
author = {Toshiaki Fujiwara and Hiroshi Fukuda and Hiroshi Ozaki},
title = {Choreographic three bodies on the lemniscate},
journal = {Journal of Physics A: Mathematical and General}
}

@article{BarutelloFerrarioTerracini2011,
  author  = {Barutello, Vivina and Ferrario, Davide and Terracini, Susanna},
  title   = {Symmetry groups of the planar $N$-body problem and action-minimizing solutions},
  journal = {Archive for Rational Mechanics and Analysis},
  volume  = {190},
  number  = {1},
  pages   = {189--237},
  year    = {2008},
  doi     = {10.1007/s00205-008-0131-7}
}

@article{Evans1990Superintegrability,
  author       = {Evans, N. W.},
  title        = {Superintegrability in Classical Mechanics},
  journal      = {Physical Review A},
  volume       = {41},
  number       = {10},
  pages        = {5666--5676},
  year         = {1990},
  doi          = {10.1103/PhysRevA.41.5666},
  publisher    = {American Physical Society}
}

@book{Perelomov1990IntegrableSystems,
  author       = {Perelomov, A. M.},
  title        = {Integrable Systems of Classical Mechanics and Lie Algebras},
  year         = {1990},
  publisher    = {Birkh\"auser},
  address      = {Basel},
  isbn         = {9783764323363},
  note         = {}
}

\clearpage

\includepdf[
  pages=-,
  fitpaper=true,
  pagecommand={\thispagestyle{empty}}
]{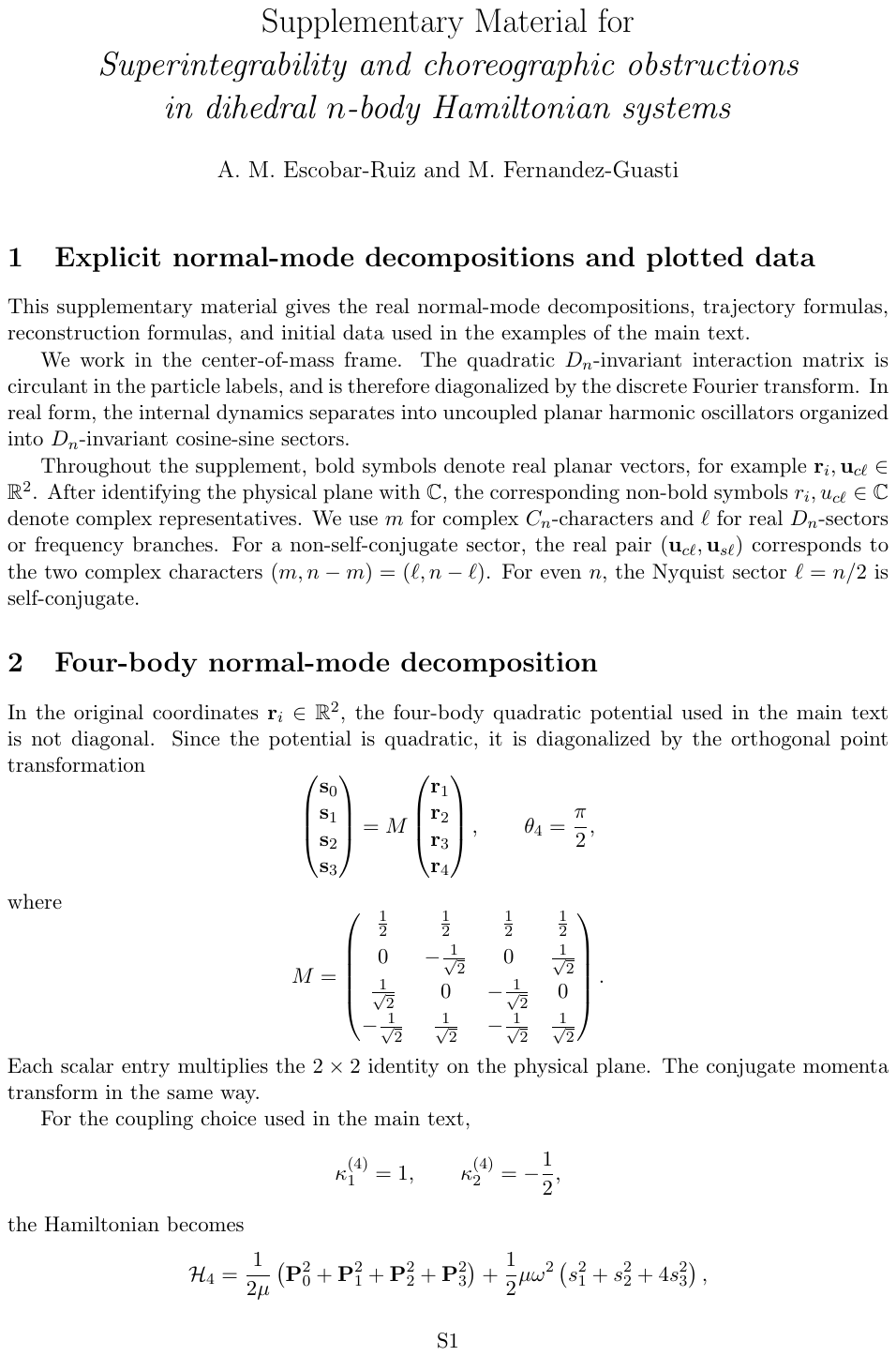}

\end{document}